\documentclass[a4paper]{article}

\usepackage{amsmath}
\usepackage{amssymb}
\usepackage{graphicx}
\usepackage[OT4]{fontenc}
\usepackage[latin2]{inputenc}
\usepackage{tikz}
\usepackage{url}
\usepackage[numbers]{natbib}
\usepackage{bm}
\usepackage{setspace}

\newcommand{\abs}[1]{\left| #1 \right|}

\newcommand{\okra}[1]{\left( #1 \right)}
\newcommand{\kwad}[1]{\left[ #1 \right]}
\newcommand{\klam}[1]{\left\{ #1 \right\}}

\newcommand{\floor}[1]{\left\lfloor #1 \right\rfloor}

\DeclareMathOperator{\esssup}{ess\, sup}
\DeclareMathOperator{\card}{card}

\DeclareMathOperator{\sred}{\mathbf{E}}

\newcommand{\boole}[1]{{\bf 1}{\klam{#1}}}

\newtheorem{theorem}{Theorem}
\newtheorem{lemma}{Lemma}
\newenvironment*{proof}{\begin{trivlist}\item[]
\noindent\textbf{Proof:}}{$\Box$\par\end{trivlist}}
\newenvironment*{proof*}[1]{\begin{trivlist}\item[]
\noindent\textbf{Proof of #1:}}{$\Box$\par\end{trivlist}}

\author{{\L}ukasz D\k{e}bowski\thanks{
    {\L}. D\k{e}bowski is with
    the Institute of Computer Science, Polish Academy of Sciences, 
    ul. Jana Kazimierza 5, 01-248 Warszawa, Poland 
    (e-mail: ldebowsk@ipipan.waw.pl).
  }
}

\title{Maximal Repetition and Zero Entropy Rate} \date{}

\begin{document}

\pagestyle{empty}   
\begin{titlepage}
\maketitle

\begin{abstract}
  Maximal repetition of a string is the maximal length of a repeated
  substring.  This paper investigates maximal repetition of strings
  drawn from stochastic processes.  Strengthening previous results,
  two new bounds for the almost sure growth rate of maximal repetition
  are identified: an upper bound in terms of conditional R\'enyi
  entropy of order $\gamma>1$ given a sufficiently long past and a
  lower bound in terms of unconditional Shannon entropy ($\gamma=1$).
  Both the upper and the lower bound can be proved using an inequality
  for the distribution of recurrence time. We also supply an
  alternative proof of the lower bound which makes use of an
  inequality for the expectation of subword complexity.  In
  particular, it is shown that a power-law logarithmic growth of
  maximal repetition with respect to the string length, recently
  observed for texts in natural language, may hold only if the
  conditional R\'enyi entropy rate given a sufficiently long past
  equals zero. According to this observation, natural language cannot
  be faithfully modeled by a typical hidden Markov process, which is a
  class of basic language models used in computational linguistics.
%
  \\[1em]
  \textbf{Keywords}: maximal repetition, R\'enyi entropies, entropy
  rate, recurrence time, subword complexity, natural language
\end{abstract}


\end{titlepage}
\pagestyle{plain}   


\section{Motivation and main results}
\label{secIntroduction}

Maximal repetition $L(x_1^n)$ of a string
$x_1^n=(x_1,x_2,...,x_n)$ is the maximal length of a repeated
substring. Put formally,
\begin{align}
  L(x_1^n):=\max\klam{k: x_{i+1}^{i+k}=x_{j+1}^{j+k} \text{ for some }
    0\le i<j\le n-k}.
\end{align}
Maximal repetition has been studied by computer scientists
\cite{DeLuca99,KolpakovKucherov99a,KolpakovKucherov99,CrochemoreIlie08},
probabilists
\cite{ErdosRenyi70,ArratiaWaterman89,Shields92b,Shields97}, and
information theorists
\cite{KontoyiannisSuhov94,Debowski11b,Debowski15g}. Maximal repetition
$L(x_1^n)$ can be computed efficiently for relatively long strings, in
time $O(n)$ \cite{KolpakovKucherov99a}, which opens way to various
empirical statistical studies. Moreover, for an arbitrary stochastic
process $(X_i)_{i=-\infty}^\infty$, maximal repetition $L(X_1^n)$ is
an nondecreasing function of the string length $n$.  In this paper, we
will investigate the rate of growth of maximal repetition for some
stochastic processes. 

Our theoretical investigations are motivated by an application to
statistical modeling of natural language.  In a previous paper of ours
\cite{Debowski15f}, we have been interested in the growth rate of
maximal repetition for texts in natural language.  Investigating 35
texts in English, French, and German, we have found that
a power-law logarithmic growth of maximal repetition,
\begin{align}
  \label{Debowski}
  L(x_1^n)\approx A\okra{ \log n }^{\alpha}
  ,
\end{align}
holds approximately with $\alpha\approx 3$. This empirical result
seems rather significant. It should be immediately noted that for a
random permutation of text characters, we observe the plain
logarithmic growth of maximal repetition,
\begin{align}
  \label{Logarithmic}
  L(x_1^n)\approx A\log n
  .
\end{align}
For a quick reference, in Figure \ref{figMaximalRepetition}, we
present the data for a collection of plays by William Shakespeare,
downloaded from Project Gutenberg (\url{http://www.gutenberg.org/}). To
smooth the plot, we have computed maximal repetition for strings
$x_{c_n+1}^{c_n+n}$ rather than $x_1^n$, where offsets $c_n$ are
selected at random.

\begin{figure}[h]
  \centering
  \includegraphics[width=0.9\textwidth]{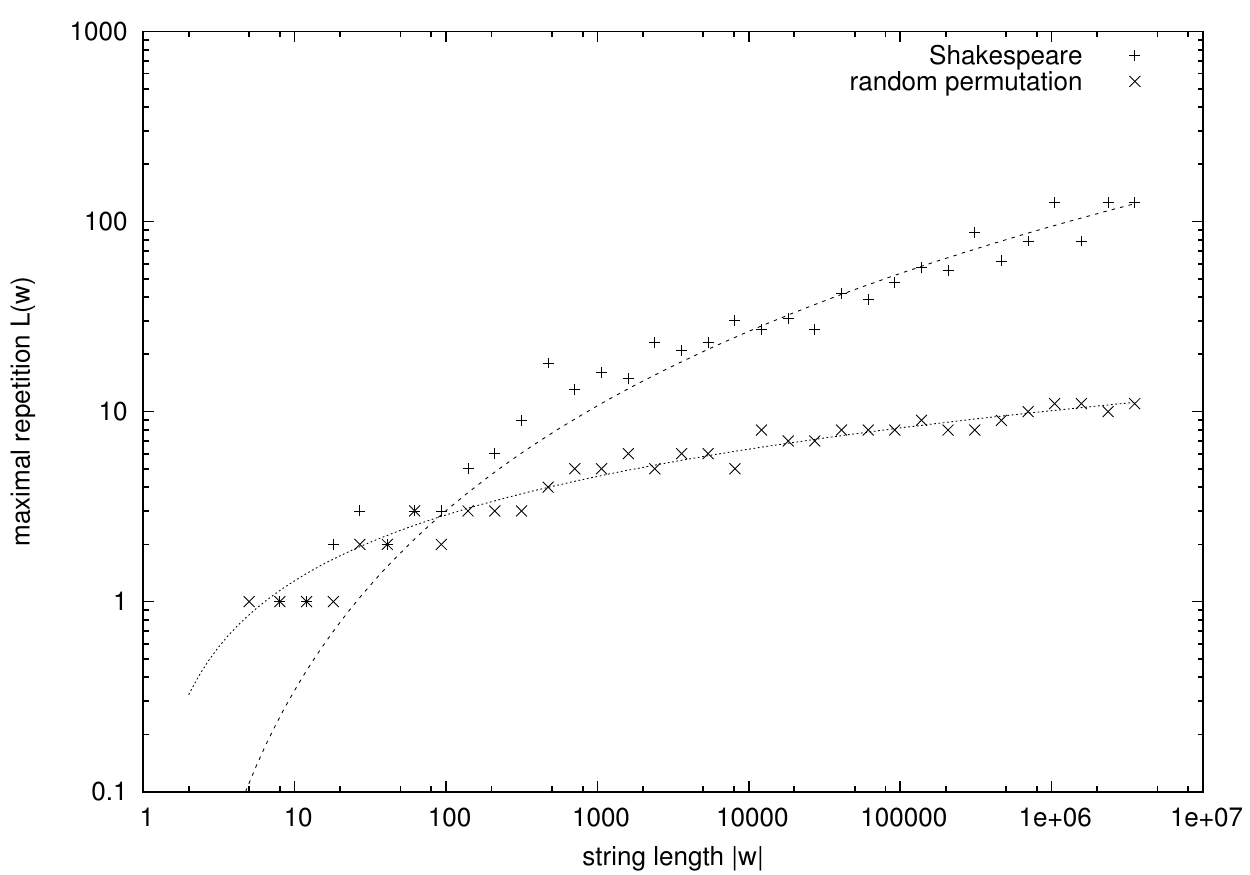}
  \caption{\label{figMaximalRepetition} Maximal repetition for the
    concatenation of 35 plays by William Shakespeare and a random
    permutation of the text characters. To smooth the results, we
    sampled substrings $w=x_{c_n+1}^{c_n+n}$ of length $\abs{w}=n$
    from both sources at random assuming a uniform probability
    distribution on $c_n=0,1,...,N-n$, where $N$ is the source
    length. For each length $\abs{w}=n$, only one substrings
    $w=x_{c_n+1}^{c_n+n}$ was sampled. The fitted model is
    $L(x_{c_n+1}^{c_n+n})\approx 0.02498 \okra{\log n}^{3.136}$ for
    Shakespeare and $L(x_{c_n+1}^{c_n+n})\approx 0.4936 \okra{\log
      n}^{1.150}$ for the random permutation.}
\end{figure}

Consequently, we may ask what the empirical law (\ref{Debowski}) can
tell us about the stochastic mechanism of natural language
generation. Let $(X_i)_{i=-\infty}^\infty$ be a stationary process.
We consider the Shannon entropy
\begin{align}
  H(n):=\sred\kwad{-\log P(X_1^n)}
\end{align}
and the associated Shannon entropy rate
\begin{align}
  h:=\lim_{n\rightarrow\infty} \frac{H(n)}{n}
  .
\end{align}
Maximal repetition $L(X_1^n)$ resembles another statistic that has
been intensely investigated, the longest match length $L_n$, which is
the maximal length $k$ such that string $X_0^{k-1}$ is a substring of
$X_{-n}^{-1}$
\cite{OrnsteinWeiss93,Szpankowski93,KontoyiannisOthers98,GaoKontoyiannisBienenstock08}. As
shown in \cite{OrnsteinWeiss93}, for a stationary ergodic process over
a finite alphabet, we have the pointwise convergence
\begin{align}
  \lim_{n\rightarrow\infty} \frac{\log n}{1+L_n}= h \text{ a.s.}
\end{align}
Since $L_n\le L(X_{-n}^{n-1})$ then, as discussed by Shields
\cite{Shields92b}, we obtain a logarithmic bound for the maximal
repetition,
\begin{align}
  \label{hmaxrep}
  \limsup_{n\rightarrow\infty} \frac{\log n}{1+L(X_{-n}^{n-1})}\le h
  \text{ a.s.}
\end{align}
Hence the growth rate of maximal repetition provides a lower bound for
the Shannon entropy rate.

Does then the power-law logarithmic growth of maximal repetition
(\ref{Debowski}) imply that the Shannon entropy rate of natural
language is zero? Here let us note that the overwhelming evidence
collected so far suggests that the Shannon entropy rate of natural
language is strictly positive, $h\approx 1$ bit per character
\cite{Shannon51,CoverKing78,BrownOthers83,Grassberger02,BehrOthers03,TakahiraOthers16}
but among researchers investigating this question there was an
exception. Namely, Hilberg \cite{Hilberg90} supposed that the Shannon
entropy of natural language satisfies condition $H(n)\approx
Bn^\beta$, where $\beta\approx 0.5$, and consequently the Shannon
entropy rate might be zero.  Although we have not been convinced that
the Shannon entropy rate of natural language equals zero, for some
time we have been interested in relaxations and strengthenings of
Hilberg's hypothesis, see
\cite{Debowski11b,Debowski15g,Debowski15f,TakahiraOthers16}. In
particular, we have been quite disturbed by the power-law logarithmic
growth of maximal repetition for natural language, which we observed
by the way. In \cite{Debowski15f}, we supposed that it should be
linked with vanishing of some sort of an entropy rate.

This hypothetical entropy rate cannot be the Shannon entropy rate,
however. As also shown by Shields \cite{Shields92b}, bound
(\ref{hmaxrep}) is not tight.  For any stationary ergodic process
$(X_i)_{i=-\infty}^\infty$ and a function $\lambda(n)=o(n)$, there is
a measurable function $f$ of infinite sequences and a stationary
ergodic process $(Y_i)_{i=-\infty}^\infty$, where
$Y_i:=f((X_{i+j})_{j=-\infty}^\infty)$, such that
\begin{align}
  \limsup_{n\rightarrow\infty} \frac{L(Y_1^n)}{\lambda(n)}\ge 1 \text{ a.s.}
\end{align}
Whereas the Shannon entropy rate of process $(Y_i)_{i=-\infty}^\infty$
is smaller than that of process $(X_i)_{i=-\infty}^\infty$, a careful
analysis of the proof shows that the difference between the two can be
made arbitrarily small, cf.\ e.g.\ \cite{Gray90}. Moreover, if we take
$(X_i)_{i=-\infty}^\infty$ to be an IID process, then process
$(Y_i)_{i=-\infty}^\infty$ is mixing and very weak Bernoulli.  Hence
the power-law logarithmic growth of maximal repetition
(\ref{Debowski}) does not imply that the Shannon entropy rate of
natural language is zero or that natural language is not mixing.

In spite of this negative result, in this article, we will show that
the power-law logarithmic growth of maximal repetition is naturally
linked to a power-law growth of some \emph{generalized} block
entropies and vanishing of some \emph{generalized} entropy rates.  For
simplicity let us consider a stationary process
$(X_i)_{i=-\infty}^\infty$.  For a parameter
$\gamma\in(0,1)\cup(1,\infty)$, the block R\'enyi entropy
\cite{Renyi61} is defined as
\begin{align}
  \label{Renyi}
  H_\gamma(n):=\frac{1}{1-\gamma}\log \sum_{x_1^n} P(X_1^n=x_1^n)^{\gamma}.
\end{align}
For $\gamma\in\klam{0,1,\infty}$, we define the block R\'enyi entropy
as
\begin{align}
  H_\gamma(n):=\lim_{\delta\rightarrow\gamma} H_\delta(n).
\end{align}
Some special cases of $H_\gamma(n)$ are:
\begin{enumerate}
\item Hartley entropy $H_0(n)=\log \card\klam{x_1^n: P(X_1^n=x_1^n)>0}$,
\item Shannon entropy $H_1(n)=H(n)=\sred\kwad{-\log P(X_1^n)}$,
\item collision entropy $H_2(n)=-\log \sred P(X_1^n)$,
\item min-entropy $H_\infty(n)=-\log \max_{x_1^n} P(X_1^n=x_1^n)$.
\end{enumerate}
We have $H_\gamma(n)\ge H_\delta(n)$ for $\gamma<\delta$ and
$H_\gamma(n)\le \frac{\gamma}{\gamma-1} H_\infty(n)$ for $\gamma>1$.

In our problem, we will also deal with some conditional R\'enyi
entropies given the infinite past. In the literature, a few
alternative definitions of conditional R\'enyi entropy have been
presented, cf.\ \cite{Arikan96,Berens13}. Here we will use yet another
definition which is algebraically simpler and arises naturally in our
application. For a parameter $\gamma\in(0,1)\cup(1,\infty)$, the
conditional block R\'enyi entropy will be defined as
\begin{align}
  \label{Renyi}
  H^{cond}_\gamma(n):=\frac{1}{1-\gamma}\log \sred \sum_{x_1^n}
  P(X_1^n=x_1^n|X_{-\infty}^0)^{\gamma}.
\end{align}
For $\gamma\in\klam{0,1,\infty}$, we define the conditional block
R\'enyi entropy as
\begin{align}
  H^{cond}_\gamma(n):=\lim_{\delta\rightarrow\gamma} H^{cond}_\delta(n).
\end{align}
We note that for $\gamma>1$ the conditional block
R\'enyi entropy can be written as
\begin{align}
  H^{cond}_\gamma(n)=-\frac{1}{\gamma-1}\log \sred
  \kwad{P(X_1^n|X_{-\infty}^0)}^{\gamma-1}
\end{align}
and hence we obtain the conditional block min-entropy 
\begin{align}
  H^{cond}_\infty(n)=-\log \esssup P(X_1^n|X_{-\infty}^0).
\end{align}
By the Jensen inequality, we have $H^{cond}_\gamma(n)\ge
H^{cond}_\delta(n)$ for $\gamma<\delta$ and $H^{cond}_\gamma(n)\le
H_\gamma(n)$.  In contrast, we need not have $H^{cond}_\gamma(n)\le
\frac{\gamma}{\gamma-1} H^{cond}_\infty(n)$ for $\gamma>1$ (consider
for instance $\gamma=2$). By another application of the Jensen
inequality and by equality $h=\sred\kwad{-\log P(X_1|X_{-\infty}^0)}$,
we obtain the chain of inequalities
\begin{align}
  H^{cond}_\infty(n)\le H_\gamma^{cond}(n)\le hn\le H(n)\le H_0(n),\quad
  \gamma>1.
\end{align}
Resuming, entropies $H_0(n)$ and $H^{cond}_\infty(n)$ are the largest
one and the smallest one of the introduced entropies, respectively.

The above definitions can be partly generalized for nonstationary
processes, as well.  For an arbitrary (possibly nonstationary) process
$(X_i)_{i=1}^\infty$ over a countable alphabet, we generalize the
definition of the block Hartley entropy as
  \begin{align}
    H_0(n):=\log\card\klam{x_1^n:P(X_{m+1}^{m+n}=x_1^n)>0 \text{ for some
        $m\ge 0$}}
  \end{align}
and the conditional block min-entropy as
 \begin{align}
   H^{cond}_\infty(n):=-\log\sup_{m\ge 0}\max_{x_1^{m+n}}
   P(X_{m+1}^{m+n}=x_{m+1}^{m+n}|X_1^m=x_1^m) .
 \end{align}
As we can check easily, entropies $H_0(n)$ and $H^{cond}_\infty(n)$
coincide with the previous definitions for a stationary process.

Concerning the links between the generalized entropies and the maximal
repetition, we will begin with two simple results which consolidate
and generalize earlier observations from
\cite{Shields97,Debowski15f}---and are stated in more generality for
nonstationary processes. Namely, we will show that entropies $H_0(n)$
and $H^{cond}_\infty(n)$ provide an inverse sandwich bound for the
investigated statistic of strings.   

The first proposition bounds the maximal repetition below with the
Hartley entropy.  The smaller is the Hartley entropy, the larger is
the maximal repetition.
\begin{theorem}[cf. \cite{Debowski15f}]
  \label{theoRepetitionHartley}
  For an arbitrary process $(X_i)_{i=1}^\infty$ over a countable
  alphabet, if 
  \begin{align}
    \label{CondHartley} 
    H_0(k) \le Bk^{\beta}
  \end{align}
  for sufficiently large $k$ for certain $B>0$ and $\beta>0$ then
  \begin{align}
    \label{RepetitionHartley}
    L(X_1^n) \ge A(\log n)^{\alpha}
  \end{align}
  for sufficiently large $n$ almost surely, for any $A<
  B^{-\alpha}$ and $\alpha=1/\beta$.
\end{theorem}
\begin{proof}
  Since the alphabet is countable, $P(X_{i+1}^{i+k})>0$ holds almost
  surely for all $0\le i<\infty$. Hence block $X_1^n$ contains almost
  surely no more than $\exp(H_0(k))$ different strings of length
  $k$. In particular if $\exp(H_0(k))<n-k+1$ then block $X_1^n$
  contains a repeat of length $k$, i.e., $L(X_1^n)\ge k$. Assume that
  $H_0(k)\le Bk^\beta$ holds for sufficiently large $k$. If we put
  $k_n= A\okra{ \log n }^{1/\beta}$ where $A< B^{-1/\beta}$ then we
  obtain $H_0(k_n)\le Bk_n^\beta<\log(n-k_n+1)$ for sufficiently large
  $n$. Hence $L(X_1^n)\ge k_n$ almost surely.
\end{proof}
Theorem \ref{theoRepetitionHartley} was proved in \cite{Debowski15f}
for stationary processes. In \cite{Debowski15g}, some stationary
processes were constructed that satisfy both condition $H_0(n)\approx
Bn^\beta$ and condition $L(X_1^n) \approx A(\log n)^{\alpha}$ for an
arbitrary $\beta$ and $\alpha=1/\beta$.

In the second proposition we will bound the maximal repetition above with the
conditional min-entropy. Before, let us make a simple observation that the
conditional min-entropy is superadditive,
\begin{align}
  H^{cond}_\infty(m+n)\ge  H^{cond}_\infty(m) +H^{cond}_\infty(n)
  .
\end{align}
Hence by the Fekete lemma, we have
\begin{align}
  \lim_{n\rightarrow\infty} \frac{H^{cond}_\infty(n)}{n}=\sup_{n\ge 0}
  \frac{H^{cond}_\infty(n)}{n}
\end{align}
and consequently this limit equals zero if and only if
$H^{cond}_\infty(n)=0$ for all $n$.  Now, we observe that the larger
is the conditional min-entropy, the smaller is the maximal repetition.
\begin{theorem}[cf. \cite{Shields97}]
  \label{theoRepetitionMinEntropy}
  For an arbitrary process $(X_i)_{i=1}^\infty$ over a countable
  alphabet, if
  \begin{align}
    \label{CondMinEntropy}
    H^{cond}_\infty(k) \ge Bk
  \end{align}
  for sufficiently large $k$ for a certain $B>0$ then
  \begin{align}
    \label{RepetitionMinEntropy}
    L(X_1^n) < A\log n
  \end{align}
  for sufficiently large $n$ almost surely, for any $A> 3B^{-1}$.
\end{theorem}
\begin{proof}
We have
\begin{align}
P(L(X_1^n)\ge k)
&= P\okra{X_{i+1}^{i+k}=X_{j+1}^{j+k} \text{ for some }
    0\le i<j\le n-k}
\nonumber\\
&\le\sum_{0\le i<j\le n-k} P(X_{i+1}^{i+k}=X_{j+1}^{j+k})
\nonumber\\
&=\sum_{0\le i<j\le n-k} \sum_{x_1^{j}}
P(X_1^{j}=x_1^{j})P(X_{j+1}^{j+k}=x_{i+1}^{i+k}|X_1^{j}=x_1^{j})
\nonumber\\
&\le
\sum_{0\le i<j\le n-k} \exp(-H^{cond}_\infty(k))\le n^2 \exp(-H^{cond}_\infty(k)).
\end{align}
Assume now that $H^{cond}_\infty(k)\ge Bk$ holds for sufficiently
large $k$. If we put $k_n=A\log n$ then we obtain
\begin{align}
  \sum_{n=1}^\infty P(L(X_1^n)\ge k_n)\le \sum_{n=1}^\infty n^2
  \exp(-H^{cond}_\infty(k_n))\le C+\sum_{n=1}^\infty n^{2-BA},
\end{align}
which is finite if $A>3B^{-1}$.  Hence by the
Borel-Cantelli lemma, we obtain that
$L(X_1^n)< A\log n$ for sufficiently large $n$ almost surely.
\end{proof}

Inequality (\ref{RepetitionMinEntropy}) was demonstrated in
\cite{Shields97}, using a somewhat complicated technique involving
source coding, for processes satisfying the equivalent finite energy
condition
\begin{align}
  \label{FiniteEnergy}
  P(X_{m+1}^{m+n}=x_{m+1}^{m+n}|X_1^m=x_1^m)\le Kc^n, 
\end{align}
where $K>0$ and $0<c<1$.  Condition (\ref{FiniteEnergy}) appears
intuitive. We would expect it from well-behaved processes.  In fact,
finite energy processes include typical hidden Markov processes,
uniformly dithered processes, processes satisfying the Doeblin
conditions, as well as nonatomic $\psi$-mixing processes.


For clarity and completeness, let us state the respective results
formally. A discrete process $(Y_i)_{i=-\infty}^\infty$ is called a
hidden Markov process if $Y_i=f(X_i)$ for a certain function $f$ and a
discrete Markov process $(X_i)_{i=-\infty}^\infty$.
\begin{theorem}
  \label{theoHMM1}
  For a stationary hidden Markov process $(Y_i)_{i=-\infty}^\infty$
  let the underlying Markov process be
  $(X_i)_{i=-\infty}^\infty$. Process $(Y_i)_{i=-\infty}^\infty$ is
  finite energy if
  \begin{align}
    \label{CondHMM1}
    c:=\sup_{y,x}P(Y_i=y|X_{i-1}=x)<1
    .
  \end{align}
\end{theorem}
\begin{proof}
  By conditional independence of $Y_{m+1}$ and $Y_{1}^{m}$ given
  $X_m$,
  \begin{align}
    &P\okra{Y_{m+1}=y_{m+1}|Y_{1}^{m}=y_{1}^{m}}
    \nonumber
    \\
    &=
    \sum_{x_m} 
    P\okra{Y_{m+1}=y_{m+1}|X_m=x_m}
    P\okra{X_m=x_m|Y_{j_{1}^{m}}=y_{1}^{m}}
    \nonumber
    \\
    &\le 
    \sum_{x_m} c
    P\okra{X_m=x_m|Y_{1}^{m}=y_{1}^{m}}
    =c
    .
  \end{align}
  Thus process $(Y_i)_{i=-\infty}^\infty$ is finite energy.
\end{proof}

Another subclass of finite energy processes are uniformly dithered
processes, which generalize a construction by Shields
\cite{Shields97}.  Let $(\mathbb{X},*)$ be a group.  A stochastic
process $(X_i)_{i=-\infty}^\infty$ over the alphabet $\mathbb{X}$ is
called uniformly dithered if it satisfies $X_i=W_i*Z_i$, where
$(W_i)_{i=-\infty}^\infty$ is an arbitrary process over the alphabet
$\mathbb{X}$ and $(Z_i)_{i=-\infty}^\infty$ is an independent IID
process with $P(Z_i=a)\le c<1$.
\begin{theorem}[\cite{Debowski15f}] 
  \label{theoUD}
  Any uniformly dithered process is a finite energy process.
\end{theorem}

Let us observe that for a stationary process, condition
(\ref{FiniteEnergy}) is equivalent to 
\begin{align}
  \label{FiniteEnergyII}
  P(X_1^n=x_1^n|X_{-\infty}^0)\le Kc^n \text{ a.s.}
\end{align}
by the martingale convergence. There are two related Doeblin
conditions
\begin{align}
  \label{DoeblinI}
  P(X_r=x_r|X_{-\infty}^0)&\ge d \text{ a.s.},
  \\
  \label{DoeblinII}
  P(X_r=x_r|X_{-\infty}^0)&\le D \text{ a.s.}
\end{align}
for some $r\ge 1$ and $0<d,D<1$,
cf. \cite{GaoKontoyiannisBienenstock08,KontoyiannisOthers98}. The
first condition, can be satisfied for a finite alphabet only.
\begin{theorem}
  \label{theoDoeblin}
  If a process assuming more than one value satisfies condition
  (\ref{DoeblinI}) then it satisfies condition
  (\ref{DoeblinII}). Moreover, if a stationary process satisfies
  condition (\ref{DoeblinII}) then it is finite energy.
\end{theorem}
\begin{proof}
  First, assume condition (\ref{DoeblinI}). Then obviously
  \begin{align}
    P(X_r=x_r|X_{-\infty}^0)=1-\sum_{x'_r\neq x_r}
    P(X_r=x'_r|X_{-\infty}^0)\le 1-d=:D ,
  \end{align}
  so we obtain condition (\ref{DoeblinII}). Next, assume condition
  (\ref{DoeblinII}). Then 
  \begin{align}
    P(X_1^r=x_1^r|X_{-\infty}^0)\le
    P(X_r=x_r|X_{-\infty}^0)
    \le D
  \end{align}
  and, by stationarity, $P(X_1^n=x_1^n|X_{-\infty}^0)\le
  D^{\floor{n/r}}\le D^{n/r-1}$, so 
  (\ref{FiniteEnergyII}) follows.
\end{proof}
Independently, in \cite{KontoyiannisSuhov94}, inequality
(\ref{RepetitionMinEntropy}) was established for stationary processes
that satisfy condition (\ref{DoeblinI}).

The last subclass of finite energy processes which we are going to
discuss are nonatomic $\psi$-mixing processes. For a stationary 
process $(X_i)_{i=-\infty}^\infty$ define 
\begin{align}
  \psi(n)=\sup_{i,j\ge 1}\esssup
  \kwad{\frac{P(X_{-j}^0,X_n^{n+i})}{P(X_{-j}^0)P(X_n^{n+i})}-1}
  .
\end{align}
The process is called $\psi$-mixing if
$\lim_{n\rightarrow\infty}\psi(n)=0$.
\begin{theorem}
  A $\psi$-mixing stationary process $(X_i)_{i=-\infty}^\infty$ is
  finite energy if
  \begin{align}
    \label{NoAtoms}
    \lim_{n\rightarrow\infty}\esssup P(X_1^n)=0.
  \end{align}
\end{theorem}
\begin{proof}
  A stationary process $(X_i)_{i=-\infty}^\infty$ has been called
  simple mixing in \cite{Ko12} if
\begin{align}
  \label{SimpleMixing}
  P(X_{-j}^0,X_n^{n+i})\le KP(X_{-j}^0)P(X_n^{n+i})
\end{align}
for all $n,i,j\ge 1$ and a $K>0$. Obviously, any $\psi$-mixing process
is simple mixing. It has been shown in \cite[Corollary 4.4]{Ko12} that
if a simple mixing process $(X_i)_{i=-\infty}^\infty$ satisfies
(\ref{NoAtoms}) then $P(X_1^n)\le c^n$ for some $0<c<1$. In
consequence, any such process has the finite energy property by
condition (\ref{SimpleMixing}).
\end{proof}


There is an interesting application of the above results to natural
language. Although hidden Markov processes are some classical models
in computational linguistics
\cite{Jelinek97,ManningSchutze99,Rosenfeld00}, their insufficiency as
models of natural language was often claimed earlier, cf.\
\cite{Rosenfeld00}. Using Theorems \ref{theoRepetitionMinEntropy},
\ref{theoHMM1}, and \ref{theoDoeblin} and the empirical observation of
the power-law logarithmic growth of maximal repetition
(\ref{Debowski}), we can provide a rigorous way of demonstrating that
natural language is not a typical hidden Markov process, cf.\ a
different approach to this question in \cite{LinTegmark17}, and does
not even satisfy the Doeblin condition, contrary to an empirically
unsupported assertion in \cite{KontoyiannisOthers98}. Simply, as we
have stated in the previous paragraph, typical hidden Markov processes
and processes satisfying the Doeblin condition are finite energy,
whereas the power-law logarithmic growth (\ref{Debowski}) by Theorem
\ref{theoRepetitionMinEntropy} excludes the class of finite energy
processes.

Let us come back to the main thread.  In view of Theorems
\ref{theoRepetitionHartley} and \ref{theoRepetitionMinEntropy}, the
hyperlogarithmic growth of maximal repetition can be connected to
vanishing of the Hartley entropy rate and the conditional min-entropy,
as follows,
\begin{align}
  \limsup_{n\rightarrow\infty} \frac{H_0(n)}{n}=0
  &\implies
  \liminf_{n\rightarrow\infty} \frac{L(X_1^n)}{\log n}=\infty \text{ a.s.}
  ,
  \\
  \limsup_{n\rightarrow\infty} \frac{L(X_1^n)}{\log n}=\infty \text{ a.s.}
  &\implies
  H^{cond}_\infty(n)=0,\quad n\ge 1
  .
\end{align}
Since the difference between $H_0(n)$ and $H^{cond}_\infty(n)$ can be
arbitrarily large, we can ask a question whether the gap between the
upper bound and the lower bound for the maximal repetition can be
narrowed. The natural step is to consider other generalized entropies.

Now we can present some strengthening of Theorems
\ref{theoRepetitionHartley} and \ref{theoRepetitionMinEntropy}, which
constitutes the main result of this article. The first proposition
bounds the maximal repetition below with the Shannon entropy.  The
smaller is the Shannon entropy, the larger is the maximal repetition.
\begin{theorem}
  \label{theoRepetitionShannon}
  For a stationary process $(X_i)_{i=-\infty}^\infty$ over a countable
  alphabet, if 
  \begin{align}
    \label{CondShannon}
    H(k) \le Bk^{\beta}
  \end{align}
  for sufficiently large $k$ for certain $B>0$ and $\beta>0$ then
  \begin{align}
    \label{RepetitionShannon}
    L(X_1^n) > (\log n)^{\alpha}
  \end{align}
  for sufficiently large $n$ almost surely, for any $\alpha<1/\beta$.
\end{theorem}

In contrast, the second proposition bounds the maximal repetition
above in terms of the conditional R\'enyi entropy of order $\gamma>1$
given a sufficiently long but finite past.  For a stationary process
$(X_i)_{i=-\infty}^\infty$ over a finite alphabet $\mathbb{X}$ and
$\gamma>1$, let us write $N(n):=(\card\mathbb{X})^n$ and
  \begin{align}
    \tilde H^{cond}_\gamma(n):=-\frac{1}{\gamma-1}\log \sred
    \kwad{P(X_1^n|X_{-N(n)}^0)}^{\gamma-1} ,
  \end{align}
  where $H_\gamma(n)\ge \tilde H^{cond}_\gamma(n)\ge
  H^{cond}_\gamma(n)$ by the Jensen inequality. The larger is the
  entropy $\tilde H^{cond}_\gamma(n)$, the smaller is the maximal
  repetition.
\begin{theorem}
  \label{theoRepetitionRenyi}
  For a stationary process $(X_i)_{i=-\infty}^\infty$ over a finite
  alphabet and a $\gamma>1$, if 
  \begin{align}
    \label{CondRenyi}
    \tilde H^{cond}_\gamma(k) \ge Bk^{\beta}
  \end{align}
  for sufficiently large $k$ for certain $B>0$ and $\beta>0$ then
  \begin{align}
    \label{RepetitionRenyi}
    L(X_1^n) < A(\log n)^{\alpha}
  \end{align}
  for sufficiently large $n$ almost surely, for any $A>
  \kwad{\gamma\cdot\frac{\gamma+1}{\gamma-1}}^{\alpha}B^{-\alpha}$ and
  $\alpha=1/\beta$.
\end{theorem}

Thus, the hyperlogarithmic growth of maximal repetition can be connected
to vanishing of the Shannon entropy rate and the conditional R\'enyi
entropy rate, as follows,
\begin{align}
  h=\lim_{n\rightarrow\infty} \frac{H(n)}{n}=0
  &\implies
  \liminf_{n\rightarrow\infty} \frac{L(X_1^n)}{(\log n)^\alpha}=\infty \text{ a.s.}
  ,\quad \alpha<1,
  \\
  \limsup_{n\rightarrow\infty} \frac{L(X_1^n)}{\log n}=\infty \text{ a.s.}
  &\implies
  \liminf_{n\rightarrow\infty} \frac{\tilde H^{cond}_\gamma(n)}{n}=0
  ,\quad \gamma>1.
\end{align}
As we have mentioned, the first implication was noticed in
\cite{Shields92b}. Theorem \ref{theoRepetitionShannon} supplements
this observation for the power-law growth of Shannon entropy.  The gap
between entropies $H(n)$ and $\tilde H^{cond}_\gamma(n)$ can be still
arbitrarily large. It remains an open question whether Theorems
\ref{theoRepetitionShannon} and \ref{theoRepetitionRenyi} can be
sharpened further. Can entropies $H(n)$ and $\tilde
H^{cond}_\gamma(n)$ be both replaced with the unconditional R\'enyi
entropy $H_\gamma(n)$ of any order $\gamma>1$?  We suppose that the answer
is negative but the counterexamples seem difficult to construct.

To conclude the introduction, a few words are due about the proofs of
our new results and their historical context. Both Theorems
\ref{theoRepetitionShannon} and \ref{theoRepetitionRenyi} can be
proved using the probabilistic upper and lower bounds for recurrence
times by Kontoyiannis \cite{Kontoyiannis98}. The recurrence times are
random distances between two occurrences of a particular string in the
realization of a stationary process.  Recurrence times are a classical
topic in ergodic theory and information theory. Their fundamental
links with probability and Shannon entropy rate have been established
in \cite{Kac47,WynerZiv89,OrnsteinWeiss93}. Less recognized are their
links with R\'enyi entropy \cite{Ko12}. Recently, recurrence times
have been also researched experimentally for natural language
\cite{AltmannPierrehumbertMotter09}. Additionally, we can supply an
alternative proof of Theorem \ref{theoRepetitionShannon} which applies
subword complexity and an inequality by D\k{e}bowski
\cite{Debowski16,Debowski16b}.  The subword complexity of a string is
a function which tells how many different substrings of a given length
appear in the string.  Subword complexity has been studied mostly from
a combinatorial perspective \cite{JansonLonardiSzpankowski04,
  Ferenczi99, DeLuca99, GheorghiciucWard07, Ivanko08}, whereas its
links with entropy have not been much researched.

The remaining parts of this paper are organized as follows.  In
Section \ref{secRenyi}, we prove Theorem \ref{theoRepetitionRenyi},
whereas in Section \ref{secShannon}, we demonstrate Theorem
\ref{theoRepetitionShannon}, whose discussion partly relies on the
discussion of Theorem \ref{theoRepetitionRenyi}.

\section{Proof of Theorem \ref{theoRepetitionRenyi}}
\label{secRenyi}

Our proof of Theorem \ref{theoRepetitionRenyi} applies the concept of
the recurrence time, which is a special case of the waiting time. The
waiting time $R(x_1^k)$ is a random variable equal to the first
position in the infinite random past
$X_{-\infty}^{-1+k}=(...,X_{-3+k},X_{-2+k},X_{-1+k})$ at which a copy
of a finite fixed string $x_1^k=(x_1,x_2,...,x_k)$ appears,
\begin{align}
  R(x_1^k):=\inf\klam{i\ge 1: X_{-i+1}^{-i+k}=x_{1}^{k}}
  .
\end{align}
A particular case of the waiting time is the recurrence time
$R_k:=R(X_1^k)$, where we plug in the random block $X_1^k$.  To bound
the maximal repetition with conditional R\'enyi entropy, we first link
the distribution of maximal repetition to the expectation of the
recurrence time.
\begin{lemma}
  \label{theoRecurrenceRepetitionProb}
  For a stationary process $(X_i)_{i=-\infty}^\infty$ over a countable
  alphabet, 
\begin{align}
  \label{MaxRepRecurrence}
  P(L(X_1^n)<k)&\le \frac{\sred \log R_k}{\log(n-k+1)},
  \\
  \label{MaxRepRecurrenceII}
  P(L(X_1^n)\ge k)&\le (n-k+1)^\gamma\sred R_k^{-\gamma+1}, \quad \gamma>1.
\end{align}
\end{lemma}
\begin{proof}
  Let $T$ be the shift operation, $X_i\circ T=X_{i+1}$. We have 
  \begin{align}
    (L(X_{-n+k-1}^k)\ge k)=\bigcup_{i=k}^{n} (R_k\circ T^{-i}\le n-i+1).
  \end{align}
  Hence by stationarity and the Markov inequality,
   \begin{align}
     P(L(X_1^n)<k)&\le P(R_k>n-k+1)\le \frac{\sred \log R_k}{\log(n-k+1)},
     \\
     P(L(X_1^n)\ge k)&\le (n-k+1)P(R_k\le n-k+1) \le 
     (n-k+1)^\gamma\sred R_k^{-\gamma+1}.
   \end{align}
\end{proof}

Now let us introduce trimmed waiting and recurrence times
\begin{align}
  S(x_1^k)&:=\min\klam{R(x_1^k),N(k)}\le R(x_1^k),
  \\
  S_k&:=\min\klam{R_k,N(k)}\le R_k,
\end{align}
where $N(k):=(\card\mathbb{X})^k$ and $\mathbb{X}$ is the alphabet of
$X_i$.  Subsequently, we have a bound for the distribution of the
trimmed recurrence time in terms of conditional probability. This
bound is inspired by a similar bound for the untrimmed recurrence time
$R_k$ given by Kontoyiannis \cite{Kontoyiannis98}. The result of
Kontoyiannis applied conditional probability given the infinite
past. Here we reduce this infinite past to a finite context. The proof
technique remains essentially the same.
\begin{lemma}[cf.\ \cite{Kontoyiannis98}]
  \label{theoProbRecurrence}
  For a process $(X_i)_{i=-\infty}^\infty$ over a finite alphabet
  $\mathbb{X}$, for any $C>0$, we have
\begin{align}
  \label{ProbRecurrence}
  P\okra{S_k\le \frac{C}{P(X_1^k|X_{-N(k)}^0)}}\le C(1+k\log\card\mathbb{X})
  .
\end{align}
\end{lemma}
\begin{proof} 
By the conditional Markov inequality, we have 
\begin{align}
  P\okra{S_k\le \frac{C}{P(X_1^k|X_{-N(k)}^0)}}
  &\le 
  \sred \sred\okra{\frac{C}{S_kP(X_1^k|X_{-N(k)}^0)}\middle|
    X_{-N(k)}^0}
  \nonumber\\
  &=
  \sred\sum_{x_1^k}\frac{C}{S(x_1^k)} 
  .
\end{align}
But for each $i\ge 1$ there is at most one string $x_1^k$ such that
$R(x_1^k)=i$, so we have a uniform almost sure bound
\begin{align}
  \sum_{x_1^k}\frac{C}{S(x_1^k)} \le
  \sum_{i=1}^{N(k)}\frac{C}{i}
  \le
  C\okra{1+\int_{1}^{N(k)}\frac{1}{u}du}
  \le C(1+k\log\card\mathbb{X})
  .
\end{align}
\end{proof}

Having demonstrated the above two lemmas, we are in a position to
prove Theorem \ref{theoRepetitionRenyi}. For $\gamma>1$, assume
$\tilde H^{cond}_\gamma(k) \ge Bk^{\beta}$ for sufficiently large
$k$. Observe that for $0\le X\le 1$ and $Y\ge 0$, we have $\sred X\le
\sred Y+P(X\ge Y)$.  Specializing this to $0\le S_k^{-1}\le 1$ and
$P(X_1^k|X_{-N(k)}^0)\ge 0$, by (\ref{ProbRecurrence}) we obtain
\begin{align}
  \sred R_k^{-\gamma+1}&\le \sred S_k^{-\gamma+1}
  \nonumber\\
  &\le \min_{C>0}\kwad{C^{-\gamma+1}\sred
    \kwad{P(X_1^k|X_{-N(k)}^0)}^{-\gamma+1} + P\okra{S_k\le
      \frac{C}{P(X_1^k|X_{-N(k)}^0)}}}
  \nonumber\\
  &\le \min_{C>0}\kwad{C^{-\gamma+1}\exp\okra{-(\gamma-1)\tilde
      H^{cond}_\gamma(k)}+Ck(\log\card\mathbb{X}+1)}
  \nonumber\\
  &\le 2\kwad{(\log\card\mathbb{X}+1)k\exp\okra{-\tilde
      H^{cond}_\gamma(k)}}^{\frac{\gamma-1}{\gamma}} .
\end{align}
Hence by (\ref{MaxRepRecurrenceII}), we obtain for sufficiently large
$k$ that
\begin{align}
  P(L(X_1^n)\ge k)&\le
  2n^\gamma\kwad{(\log\card\mathbb{X}+1)k\exp\okra{-\tilde
      H^{cond}_\gamma(k)}}^{\frac{\gamma-1}{\gamma}} .
\end{align}
Let us take $k_n=A(\log n)^{1/\beta}$. Then
\begin{align}
    \sum_{n=1}^\infty P(L(X_1^n)\ge k_n)\le C+ 2\sum_{n=1}^\infty
    n^\gamma \kwad{(\log\card\mathbb{X}+1)A(\log n)^{1/\beta}n^{-BA^\beta}}^{\frac{\gamma-1}{\gamma}},
  \end{align}
  which is finite for
  $A>\kwad{\gamma\cdot\frac{\gamma+1}{\gamma-1}}^{1/\beta}B^{-1/\beta}$.
  Hence by the Borel-Cantelli lemma, we have $L(X_{1}^{n})< A(\log
  n)^{1/\beta}$ for all but finitely many $n$ almost surely. In this
  way we have proved Theorem \ref{theoRepetitionRenyi}.

\section{Two proofs of Theorem \ref{theoRepetitionShannon}}
\label{secShannon}

We will present two proofs of Theorem \ref{theoRepetitionShannon}. The
first one uses inequality (\ref{MaxRepRecurrence}) and the second
bound for the recurrence time by Kontoyiannis \cite{Kontoyiannis98}.
\begin{lemma}[cf.\ \cite{Kontoyiannis98}]
  \label{theoProbRecurrenceII}
  For a stationary process $(X_i)_{i=-\infty}^\infty$ over a countable
  alphabet, for any $C>0$, we have
\begin{align}
  \label{ProbRecurrenceII}
  P\okra{R_k\ge \frac{C}{P(X_1^k)}}\le C^{-1}
  .
\end{align}
\end{lemma}
\begin{proof} 
  Recalling the celebrated Kac theorem,
\begin{align}
  \label{Kac}
  \sred(R_k|X_1^k)=\frac{1}{P(X_1^k)},
\end{align}
cf.\ \cite{Kac47}, we obtain by the conditional Markov inequality
\begin{align}
  P\okra{R_k\ge \frac{C}{P(X_1^k)}}
  &=
  \sred P\okra{R_k\ge \frac{C}{P(X_1^k)}\middle|X_1^k}
  \nonumber\\
  &\le \sred\kwad{\frac{P(X_1^k)\sred(R_k|X_1^k)}{C}}=C^{-1}
  .
\end{align}
\end{proof}

Observe that $\sred X\le \int_0^\infty P(X\ge p)dp$. Hence,
applying Lemma \ref{theoProbRecurrenceII} to inequality
(\ref{MaxRepRecurrence}), we obtain
\begin{align}
  P(L(X_1^n)<k)&\le \frac{\sred \log R_k}{\log(n-k+1)}
  \nonumber\\
  &= \frac{\sred\kwad{-\log P(X_1^k)}+\sred\kwad{\log R_k
      P(X_1^k)}}{\log(n-k+1)}
  \nonumber\\
  &\le \frac{H(k)+\int_0^\infty P\okra{R_kP(X_1^k)\ge
      e^{p}}dp}{\log(n-k+1)}
  \nonumber\\
  &\le \frac{H(k)+\int_0^\infty e^{-p}dp}{\log(n-k+1)} =
  \frac{H(k)+1}{\log(n-k+1)} .
  \label{MaxRepRecurrenceHI}
\end{align}

Assume $H(k)\le Bk^{\beta}$ for sufficiently large $k$. For an $l>0$
and an $\epsilon>0$, let us take $k_l=2^l$, $m=l^{1+\epsilon}$, and
$n_l=\exp(k_l^{\beta+\epsilon})$. By inequality
(\ref{MaxRepRecurrenceHI}) for $k=k_l$, $m=m_l$, and $n=n_l$, we obtain
  \begin{align}
    \sum_{l=1}^\infty P(L(X_1^{n_l})<k_l)\le C+ \sum_{l=1}^\infty
    \frac{B2^{l\beta}+1}{\log(\exp(2^{l(\beta+\epsilon)})-2^l+1)}
    <\infty .
  \end{align}
  Hence by the Borel-Cantelli lemma, we have $L(X_1^n)\ge k_l$ for all
  but finitely many $l$ almost surely. Now for an arbitrary $n>0$, let
  us consider the maximal $l$ such that $n_l\le n$. We have
  $L(X_1^n)\ge L(X_{1}^{n_l})$ since $X_{1}^{n_l}$ is a substring of
  $X_{1}^{n}$, whereas
  \begin{align}
    k_{l}=\frac{1}{2}k_{l+1}=\frac{1}{2}(\log
    n_{l+1})^{1/(\beta+\epsilon)}\ge \frac{1}{2}(\log
    n)^{1/(\beta+\epsilon)}
    .
  \end{align}
  Hence $L(X_{1}^{n})\ge \frac{1}{2}(\log n)^{1/(\beta+\epsilon)}$
  holds for all but finitely many $n$ almost surely.  This completes
  the first proof of Theorem \ref{theoRepetitionShannon}.

The second proof of Theorem \ref{theoRepetitionShannon} will make make
use of another concept, namely, the notion of subword complexity.
Subword complexity $f(k|x_1^n)$ is a function which counts how many
distinct substrings of length $k$ appear in a string $x_1^n$,
\begin{align}
  f(k|x_1^n):=\card\klam{y_1^k: x_{i+1}^{i+k}=y_1^k \text{ for some } 0\le
    i\le n-k}.
\end{align}
To bound the maximal repetition in terms of Shannon entropy, we first
relate the distribution of maximal repetition to the expected subword
complexity. The following proposition strengthens Lemma
\ref{theoRecurrenceRepetitionProb}.
\begin{lemma}
  \label{theoTopEntropyRepetitionProb}
  We have
\begin{align}
  \label{MaxRepSubwordComp}
  P(L(X_1^n)<k)&\le \frac{\sred f(k|X_1^n)}{n-k+1},
  \\
  \label{MaxRepSubwordCompII}
  P(L(X_1^n)\ge k)&\le n-k+1-\sred f(k|X_1^n).
\end{align}
\end{lemma}
\begin{proof}
   We have $f(k|X_1^n)=n-k+1$ if $L(X_1^n)<k$ and $f(k|X_1^n)\le n-k$
   if $L(X_1^n)\ge k$. Hence
   \begin{align}
     \sred f(k|X_1^n) &\ge (n-k+1)P(L(X_1^n)<k),
     \\
     \sred f(k|X_1^n) &\le
     (n-k+1)P(L(X_1^n)<k)+(n-k)P(L(X_1^n)\ge k),
   \end{align}
   from which the claims follow.
\end{proof}

Subsequently, we have a bound for the expected subword complexity in
terms of Shannon entropy. The following Lemma \ref{theoMaxRepEntropy}
is a variation of the results in
\cite{Debowski16,Debowski16b}. Precisely, in \cite{Debowski16} we have
established inequality (\ref{SubwordCompSigma}), whereas in
\cite{Debowski16b} we have given a bound similar to
(\ref{MaxRepEntropy}) but for the number of nonoverlapping blocks
rather than the overlapping ones.
\begin{lemma}[cf.\ \cite{Debowski16,Debowski16b}]
  \label{theoMaxRepEntropy}
  For a stationary process $(X_i)_{i=-\infty}^\infty$ over a countable
  alphabet, for any $m\ge 1$,
  \begin{align}
    \label{MaxRepEntropy}
     \frac{\sred f(k|X_1^n)}{n-k+1}\le \frac{1}{m} + 
     \frac{\exp\okra{mH(k)}}{n-k+1}
  .
  \end{align}
\end{lemma}
\begin{proof}
  We will use the identity
\begin{align}
  f(k|X_1^n)=\sum_{w\in\mathbb{X}^k}\boole{\sum_{i=0}^{n-k}
    \boole{X_{i+1}^{i+k}=w}\ge 1}.
\end{align}
Hence by the Markov inequality,
\begin{align}
  \sred f(k|X_1^n)
  &= \sum_{w\in\mathbb{X}^k} P\okra{\sum_{i=0}^{n-k}
    \boole{X_{i+1}^{i+k}=w}\ge 1}
  \nonumber\\
  &\le \sum_{w\in\mathbb{X}^k} \min\kwad{1,\sred\okra{\sum_{i=0}^{n-k}
    \boole{X_{i+1}^{i+k}=w}}}
  \nonumber\\
  &= \sum_{w\in\mathbb{X}^k}
  \min\kwad{1,(n-k+1)P(X_1^k=w)}
  \nonumber\\
  &=
  (n-k+1)\sred\okra{\min\okra{\kwad{(n-k+1)P(X_1^k)}^{-1},1}}
  .
\end{align}

Denoting $\sigma(y)=\min\kwad{\exp(y),1}$, we obtain
\begin{align}
  \label{SubwordCompSigma}
  \frac{\sred f(k|X_1^n)}{n-k+1}\le\sred\sigma\okra{-\log
    P(X_1^k)-\log(n-k+1)}
  ,
\end{align}
where $\sred\kwad{-\log P(X_1^k)}=H(k)$. Therefore, using the
Markov inequality
\begin{align}
  P\okra{-\log P(X_1^k)\ge mH(k)}\le
  \frac{1}{m}
\end{align}
for $m\ge 1$, we further obtain from (\ref{SubwordCompSigma}) that
\begin{align}
  \label{SubwordCompEntropy}
  \frac{\sred f(k|X_1^n)}{n-k+1}\le \frac{1}{m} + 
  \frac{\exp\okra{mH(k)}}{n-k+1}
  .
\end{align}
Inserting (\ref{SubwordCompEntropy}) into (\ref{MaxRepSubwordComp})
yields the requested bound.
\end{proof}

The above two lemmas will be used now to demonstrate Theorem
\ref{theoRepetitionShannon}.  Chaining inequalities
(\ref{MaxRepSubwordComp}) and (\ref{MaxRepEntropy}), we obtain
inequality
\begin{align}
  \label{MaxRepRecurrenceHII}
  P(L(X_1^n)<k)&\le \frac{1}{m} + 
     \frac{\exp\okra{mH(k)}}{n-k+1},
\end{align}
which resembles inequality (\ref{MaxRepRecurrenceHI}). The sequel is
essentially the same.  Assume $H(k)\le Bk^{\beta}$ for sufficiently
large $k$. For an $l>0$ and an $\epsilon>0$, let us take $k_l=2^l$,
$m=l^{1+\epsilon}$, and $n_l=\exp(k_l^{\beta+\epsilon})$. By
inequality (\ref{MaxRepRecurrenceHII}) for $k=k_l$, $m=m_l$, and
$n=n_l$, we obtain
  \begin{align}
    \sum_{l=1}^\infty P(L(X_1^{n_l})<k_l)\le C+ \sum_{l=1}^\infty
    \kwad{\frac{1}{l^{1+\epsilon}} +
      \frac{\exp(l^{1+\epsilon}B2^{l\beta})}{\exp(2^{l(\beta+\epsilon)})-2^l+1}}
    <\infty .
  \end{align}
  Hence by the Borel-Cantelli lemma, we have $L(X_1^n)\ge k_l$ for all
  but finitely many $l$ almost surely. Now for an arbitrary $n>0$, let
  us consider the maximal $l$ such that $n_l\le n$. We have
  $L(X_1^n)\ge L(X_{1}^{n_l})$ since $X_{1}^{n_l}$ is a substring of
  $X_{1}^{n}$, whereas
  \begin{align}
    k_{l}=\frac{1}{2}k_{l+1}=\frac{1}{2}(\log
    n_{l+1})^{1/(\beta+\epsilon)}\ge \frac{1}{2}(\log
    n)^{1/(\beta+\epsilon)}
    .
  \end{align}
  Hence $L(X_{1}^{n})\ge \frac{1}{2}(\log n)^{1/(\beta+\epsilon)}$
  holds for all but finitely many $n$ almost surely.  This completes
  the second proof of Theorem \ref{theoRepetitionShannon}.

\section*{Acknowledgment}

The author wishes to thank Jan Mielniczuk, Pawe{\l} Teisseyre, Ioannis
Kontoyiannis, and anonymous reviewers for very helpful comments.

\bibliographystyle{IEEEtran}

\bibliography{0-journals-abbrv,0-publishers-abbrv,ai,mine,tcs,ql,nlp,books}

\end{document}